\documentclass [12pt]{article}

\usepackage [english]{mskrzypczak}

\usepackage{mathtools}
\usepackage{listings}

\title{B\"uchi VASS recognise $\w$\=/languages that are $\asigma{1}$\=/complete}

\ident{\init}{I}
\newcommand{\trans}[1]{\xrightarrow{#1}}

\author{\myname }

\affil{University of Warsaw}

\newcommand{\verA}{\mathnm{v}}
\newcommand{\verB}{\mathnm{x}}

\newcommand{\finA}{\mathnm{u}}

\newcommand{\orrA}{\mathnm{o}}

\newcommand{\infA}{\mathnm{\alpha}}
\newcommand{\infB}{\mathnm{\beta}}

\newcommand{\runA}{\mathnm{\rho}}

\ident{\IF}{IF}
\ident{\IFpre}{\IF_{\mathrm{pre}}}
\ident{\IFinx}{\IF_{\mathrm{inf}}}

\begin{document}
\maketitle

\begin{abstract}
This short note exhibits an example of a $\asigma{1}$-complete language that can be recognised by a one blind counter B\"uchi automaton (or equivalently a B\"uchi VASS with only one place).
\end{abstract}

In this work we study the topological complexity of $\w$\=/languages recognised by vector addition systems with states equipped with B\"uchi acceptance condition (B\"uchi VASS). This model corresponds, from the automata theoretic side, to the so\=/called partially blind multi\=/counter B\"uchi automata. As noted in~\cite{finkel_top_det_petri}, this model is able to define languages that are $\bsigma{3}$\=/complete, thus topologically harder than the languages definable by the deterministic variant of these automata. However, this topological complexity does not rule out the possibility of having a model of deterministic automata with a more complex acceptance condition that would be able to capture the expressive power of the non\=/deterministic ones.

The main result of this work shows that no model of deterministic machines with a reasonable acceptance condition can capture the expressive power of the non\=/deterministic devices.

\begin{theorem}
There exists an $\w$\=/language that is recognised by a B\"uchi VASS with one counter (i.e.~with one place) that recognises a $\asigma{1}$\=/complete $\w$\=/language.
\end{theorem}

The crucial difficulty in proving this result is the fact, that if a B\"uchi VASS has only one counter, then its number of states bounds the maximal size of an anti\=/chain of it's configurations: every two configurations with the same state are comparable with respect to the natural simulation order: if $c<c'$ then the configuration $(q,c')$ can simulate all the behaviours from the configuration $(q,c)$.

To simplify the presentation of the proof it is performed in three steps: first we provide an easy example of a $\asigma{1}$\=/complete $\w$\=/language recognised by a B\"uchi VASS with two counters; then we characterise a specific $\asigma{1}$\=/complete language (namely $\IFinx$); and finally we reduce the language $\IFinx$ to an $\w$\=/language recognised by a B\"uchi VASS with only one counter.

By the sole definition of the considered models (i.e.~non\=/deterministic partially blind multi\=/counter B\"uchi automata and B\"uchi VASS) all the $\w$\=/languages recognisable by these models belong to $\asigma{1}$. Thus, our efforts focus on providing hardness of the presented languages.

\section{Preliminary notions}

\paragraph*{B\"uchi VASS} A B\"uchi VASS (or shortly VASS, as we consider only the B\"uchi acceptance condition) is a tuple $\Aa=\langle A,Q, q_\init, F, C, \delta \rangle$, where:
\begin{itemize}
\item $A$ is a finite \emph{input alphabet},
\item $Q$ is a finite set of \emph{states},
\item $q_\init\in Q$ is the \emph{initial state},
\item $F\subseteq Q$ is the set of \emph{accepting states},
\item $C$ is a finite set of \emph{counters},
\item $\delta$ is a finite \emph{transition relation}, its elements are \emph{transitions} $(q,a,\tau,q')$ where $q,q'\in Q$, $a\in A$, and $\fun{\tau}{C}{\Z}$.
\end{itemize}
Without loss of generality we assume that the set of counters $C$ has the form $C=\{1,2,\ldots,k\}$ for some $k$ (in this work $1$ or $2$). We visually represent a transition $(q,a,\tau,q')$ by $q\trans{a:\big(\tau(1),\tau(2),\ldots,\tau(k)\big)}q'$. We say that such a transition is \emph{over} the letter $a$. If $A'\subseteq A$ then $q\trans{A':\big(\tau(1),\tau(2),\ldots,\tau(k)\big)}q'$ means that for each $a\in A'$ there is a respective transition. Similarly, $q\tran{a}q'$ and $q\trans{A'}q'$ denote the respective transitions that do not modify the counter values (i.e.~$\tau$ is constant $0$).

A \emph{configuration} of a VASS $\Aa$ is a tuple $(q,c_1,c_2,\ldots,c_k)$ where $q\in Q$, $c_1,\ldots,c_k\in \N$, and $\{1,\ldots,k\}=C$. The \emph{initial configuration} is $(q_\init,0,\ldots,0)$. We say that a transition $q\trans{a:\big(\tau(1),\ldots,\tau(k)\big)}q'$ \emph{goes} from a configuration $(q,c_1,\ldots,c_k)$ to a configuration $\big(q',c_1+\tau(1),\ldots,c_k+\tau(k)\big)$ (note that by the definition it requires all the numbers $c_i+\tau(i)$ to be non\=/negative).

Let $\infA\in A^{\w}$ be an $\w$\=/word over the input alphabet. A \emph{run} of a VASS $\Aa$ over $\infA$ is an infinite sequence $\runA$ of configurations, such that $\runA(0)$ is the initial configuration and for every $i\in\w$ there is a transition of $\Aa$ over the letter $\infA(i)$ that goes from the configuration $\runA(i)$ to the configuration $\runA(i+1)$. A run $\runA$ is \emph{accepting} if for infinitely many $i$ the configuration $\runA(i)=(q_i,\ldots)$ satisfies $q_i\in F$ (i.e.~it visits infinitely many times an accepting state). A VASS $\Aa$ \emph{accepts} an $\w$\=/word $\infA$ if there exists an accepting run of $\Aa$ over $\infA$. The \emph{language} of $\Aa$ (denoted $\lang(\Aa)$) is the set of $\w$\=/words accepted by $\Aa$.

\paragraph*{Topology} We will use the natural notions of topology on Polish spaces. By $\asigma{1}$ we denote the family of analytic sets, i.e.~projections of Borel sets.

\paragraph*{Binary trees} The binary tree is the set of all sequences of directions $\Tt\eqdef\{\dL,\dR\}^\ast$ where the \emph{directions} $\dL,\dR$ are two fixed distinct symbols. For technical reasons we sometimes consider a third direction $\dM$ (it does not occur in the binary tree).

A set $X\subseteq \Tt$ can be naturally identified with its characteristic function $X\in \{0,1\}^{\big(\{\dL,\dR\}^\ast\big)}$. Thus, the set of all subsets of the binary tree with the natural product topology is homeomorphic with the Cantor set $\{0,1\}^\omega$.

The elements $\verA,\verB\in\Tt$ are called \emph{nodes}. Nodes are naturally ordered by the following three orders:
\begin{itemize}
\item the prefix order: $\verA\preceq \verB$ if $\verB$ can be obtained by concatenating something at the end of $\verA$,
\item the lexicographic order: $\verA\lexeq \verB$ if $\verA$ is lexicographically smaller than $\verB$ (we assume that $\dL\lex \dM \lex \dR$),
\item the infix order: $\verA\inxeq \verB$ if $\verA \dM^\omega$ is lexicographically smaller than $\verB \dM^\omega$.
\end{itemize}

Notice that, for every fixed $n$, when restricted to $\{\dL,\dR\}^n$, the lexicographic and infix orders coincide. However, $\dL \inx \epsilon\inx \dR$ but $\epsilon$ is the minimal element of ${\lexeq}$. Both the lexicographic and infix orders are linear. 

\begin{fact}
The order $(\Tt,{\inxeq})$ is isomorphic with the order $(\Q,{\leq})$.
\end{fact}

\paragraph*{Hardness} In the following part of the paper we will use the following two sets:
\begin{align*}
\IFpre &\eqdef \{X\subseteq\Tt\mid \text{$X$ contains an infinite ${\preceq}$\=/ascending chain}\},\\
\IFinx &\eqdef \{X\subseteq\Tt\mid \text{$X$ contains an infinite ${\inxeq}$\=/descending chain}\}.
\end{align*}

\begin{lemma}
The sets $\IFpre$ and $\IFinx$ are $\asigma{1}$\=/complete.
\end{lemma}

\begin{proof}
Both sets belong to $\asigma{1}$ just by the form of the definition. $\IFpre$ is $\asigma{1}$\=/hard by an easy reduction from the set of ill\=/founded $\w$\=/branching trees, the proof is similar to~\cite[Exercise~27.3]{kechris_descriptive}.

$\IFinx$ is $\asigma{1}$\=/hard by a reduction from the set of ill\=/founded linear orders on $\w$ (seen as elements of $\{0,1\}^{\w\times\w}$). Let us prove this fact more formally. Consider an element $\orrA\in \{0,1\}^{\w\times\w}$ that is a linear order on $\w$ (i.e.~$(n,k)\in \orrA$ means that $n$ is $\orrA$-smaller-or-equal $k$). We will inductively define $X_{\orrA}\subseteq\Tt$ in such a way to ensure that $\orrA\mapsto X_{\orrA}$ is a continuous mapping and $\orrA$ is ill\=/founded if and only if $X_{\orrA}\in\IFinx$.

Let us proceed inductively, defining a sequence of nodes $(x_n)_{n\in\w}\subseteq\Tt$. Our invariant says that $|x_k|=k$ and the map $k\mapsto x_k$ is an isomorphism of the orders $\big(\{0,1,\ldots,n\}, \orrA\big)$ and $\big(\{x_0,x_1,\ldots,x_n\},{\inxeq}\big)$. We start with $x_0=\epsilon$ (i.e.~the root of $\Tt$). Assume that $x_0, \ldots x_n$ are defined and satisfy the invariants. By the definition of ${\inxeq}$, there exists a node $x\in\{\dL,\dR\}^n$ such that $x\inxeq x_k$ if and only if $(n{+}1,k)\in o$, for $k=0,1,\ldots,n$. Let $x_{n+1}$ be such a node.

The above induction defines an infinite sequence of nodes $x_0,x_1,\ldots$ Let $X_{\orrA}\eqdef\{x_n\mid n\in\w\}\subseteq\Tt$. By the definition of $X_{\orrA}$ the mapping $\orrA\mapsto X_{\orrA}$ is continuous --- the fact whether a node $x\in\Tt$ belongs to $X_{\orrA}$ depends only on $\orrA\cap\{0,1,\ldots,|x|\}^2$. Using our invariant, we know that the map $k\mapsto x_k$ is an isomorphism of the orders $\big(\w, \orrA\big)$ and $\big(X_{\orrA}, {\inxeq}\big)$. Thus, $\orrA$ is ill\=/founded if and only if $X_{\orrA}\in \IFinx$.
\end{proof}

\section{Hardness for \texorpdfstring{$2$\=/}{2-}counters}

In this section we provide an example of an $\w$\=/language that is $\asigma{1}$\=/complete and can be recognised by a VASS $\Aa_2$ with two counters. The VASS $\Aa_2$ is depicted on Figure~\ref{fig:two-vass}. Let $A_0\eqdef\big\{{<}, d_1, d_2, {|}, i_1, i_2, {+}, {-}, {>}\big\}$ and let the alphabet $A\eqdef A_0\cup\{\sharp\}$. The initial state is $q_0$, the only accepting state is $q_a$. The only non\=/determinism occurs in $q_0$ when reading ${<}$ --- the VASS can stay in $q_0$ or move to $q_1$. The only states that modify the counter values are $q_1$ and $q_2$.

\begin{figure}[th]
\centering
\begin{tikzpicture}[->,>=latex,auto,initial text={},scale=0.9]
\tikzstyle{trans}=[scale=0.8]
\node[state,initial  ] (q0) at (0,0) {$q_0$};
\node[state          ] (q1) at (3,0) {$q_1$};
\node[state          ] (q2) at (6,0) {$q_2$};
\node[state,accepting] (qf) at (9,+1) {$q_a$};
\node[state          ] (qr) at (9,-1) {$q_r$};
\node[state          ] (q3) at (12,0) {$q_3$};

\path (q0) edge [loop, in=60 , out=120, looseness=8] node[trans] {$A_0$} (q0)
           edge node[trans] {${<}$} (q1)
      (q1) edge [loop, in=60 , out=120, looseness=8] node[trans] {$d_1:(-1,0)$} (q1)
           edge [loop, in=240, out=300, looseness=8] node[trans] {$d_2:(0,-1)$} (q1)
           edge node[trans] {${|}$} (q2)
      (q2) edge [loop, in=60 , out=120, looseness=8] node[trans] {$i_1:(+1,0)$} (q2)
           edge [loop, in=240, out=300, looseness=8] node[trans] {$i_2:(0,+1)$} (q2)
           edge node[above, trans] {${+}$} (qf)
           edge node[below, trans] {${-}$} (qr)
      (qf) edge node[above, trans] {${>}$} (q3)
      (qr) edge node[below, trans] {${>}$} (q3)
      (q3) edge [loop, in=60 , out=120, looseness=8] node[trans] {$A_0$} (q3)
           edge [in=-90, out=270, looseness=0.8] node[below, trans] {$\sharp$} (q0);

\end{tikzpicture}
\caption{The VASS $\Aa_2$ with two counters that recognises a $\asigma{1}$\=/complete $\w$\=/language.}
\label{fig:two-vass}
\end{figure}

\begin{lemma}
\label{lem:hardness-two}
There exists a continuous reduction from $\IFpre$ to the $\w$\=/language recognised by $\Aa_2$.
\end{lemma}

\paragraph*{Intuition} An $\w$\=/word accepted by $\Aa_2$ consists of infinitely many \emph{phases} separated by $\sharp$. Each phase is a finite word over the alphabet $A_0$. In our reduction we will restrict to phases being sequences of \emph{blocks}, each block of the form given by the following definition (for $n_1,n_2,m_1,m_2\in\N$ and $s\in\{{+},{-}\}$):
\begin{equation}
B^s(-n_1,-n_2,+m_1,+m_2)\eqdef {<}\ d_1^{n_1}\ d_2^{n_2}\ {|}\ i_1^{m_1}\ i_2^{m_2}\ s\ {>}\ \in A_0^\ast.
\label{eq:block-two}
\end{equation}
Such a block is \emph{accepting} if $s={+}$, otherwise $s={-}$ and the block is \emph{rejecting}. If $\Aa_2$ starts reading a block and moves from $q_0$ to $q_1$ over ${<}$ then we say that it \emph{chooses} this block. Otherwise $\Aa_2$ stays in $q_0$ and it does not choose the given block. By the construction of the VASS $\Aa_2$, in every run it needs to choose exactly one block from each phase. Additionally, the run is accepting if and only if infinitely many of the chosen blocks are accepting.

In our reduction we will represent a given set $X\subseteq \Tt$ by an appropriately defined sequence of phases. We will control the set of configurations the VASS can reach at the beginning of each phase. These configurations will form an anti\=/chain with respect to the coordinate\=/wise order: if the VASS can reach two distinct configurations $(q_0,c_1,c_2)$ and $(q_0,c_1',c_2')$ then either $c_1<c_1'$ and $c_2>c_2'$; or $c_1>c_1'$ and $c_2<c_2'$. Each block in the successive phase will be of the form $B^s(-c_1,-c_2,+m_1,+m_2)$ for some reachable configuration $(q_0,c_1,c_2)$ --- this will be the only reachable configuration in which the automaton can choose the considered block. After choosing it, the automaton will finish reading the phase in the configuration $(q_3,m_1,m_2)$.

\paragraph*{Proof of Lemma~\ref{lem:hardness-two}} For the rest of this section we prove Lemma~\ref{lem:hardness-two}. Let us fix a set $X\subseteq \Tt$. We will construct an $\w$\=/word $\infA(X)\in A^\omega$. The $\w$\=/word $\infA(X)$ will consist of infinitely many phases $\infA(X)=\finA_0\sharp\finA_1\sharp\cdots$, for $\finA_n\in A_0^\ast$. The $n$\=/th phase $\finA_n$ (for $n=0,1,\ldots$) will depend on $X\cap\{\dL,\dR\}^{n}$. This will guarantee that the function $\fun{\infA}{2^{\Tt}}{A^\omega}$ is continuous. The proof will be concluded be the following claim.

\begin{claim}
\label{cla:red-two}
$X$ has an infinite ${\preceq}$\=/ascending chain if and only if $\Aa_2$ accepts $\infA(X)$.
\end{claim}

To simplify the construction, let us define inductively the function $\fun{b}{\Tt}{\N}$, assigning to nodes $\verA\in\Tt$ their binary value $b(\verA)$:
\begin{itemize}
\item $b(\epsilon)=0$,
\item $b(\verA \dL) = 2\cdot b(\verA)$,
\item $b(\verA \dR) = 2\cdot b(\verA)+1$.
\end{itemize}
Let $b'(\verA) = 2^n-b(\verA)-1$ for $n=|\verA|$ (i.e.~$\verA\in\{\dL,\dR\}^n$). Note that for every $n\in\N$ we have
\[b\big(\{\dL,\dR\}^n\big)=b'\big(\{\dL,\dR\}^n\big)=\{0,1,\ldots,2^n-1\},\]
and both $b$ and $b'$ are bijective between these sets. Additionally, if $\verA\neq\verA'\in\{\dL,\dR\}^n$ then either $b(\verA)<b(\verA')$ and $b'(\verA)>b'(\verA')$; or $b(\verA)>b(\verA')$ and $b'(\verA)<b'(\verA')$.

We take any $n=0,1,\ldots$ and define the $n$\=/th phase $\finA_n$. Let $\finA_n$ be the concatenation of the following blocks, for all $\verA\in\{\dL,\dR\}^n$ and $d\in\{\dL,\dR\}$:
\[B^s\big({-}b(\verA), {-}b'(\verA), {+}b(\verA d), {+}b'(\verA d)\big),\]
where $s={+}$ if $\verA\in X$ and $s={-}$ otherwise. Thus, the $n$\=/th phase is a concatenation of $2^{n+1}$ blocks, one for each node $\verA d$ in $\{\dL,\dR\}^{n+1}$.

To prove Claim~\ref{cla:red-two} it is enough to notice the following fact.

\begin{fact}
There is a bijection between infinite branches $\infB\in\{\dL,\dR\}^\omega$ and runs $\runA$ of $\Aa_2$ over $\infA(X)$. The bijection satisfies that the configuration in $\runA$ before reading the $n$\=/th phase of $\infA(X)$ is $\big(q_0, b(\verA_n), b'(\verA_n)\big)$ for $\verA_n=\infB\restr_n\in\{\dL,\dR\}^n$. $\Aa_2$ visits an accepting state in $\runA$ while reading the $n$\=/th phase of $\infA(X)$ if and only if $\verA_n\in X$.
\end{fact}

\begin{proof}
Easy induction.
\end{proof}

This concludes the proof of Lemma~\ref{lem:hardness-two}.

\section{Representation of \texorpdfstring{$\IFpre$}{IFpre}}

To construct our continuous reduction in the one\=/counter case, we need the following simple lemma that provides an alternative characterisation of the set $\IFinx$. Let us introduce the following definition.

\begin{definition}
A sequence $\verA_0,\verA_1\ldots\in \Tt$ is called a \emph{correct chain} if $\verA_0=\epsilon$ and for every $n=0,1,\ldots$:
\begin{enumerate}
\item $|\verA_{n+1}|=|\verA_n|+1$,
\label{it:cor-ch-len}
\item $\verA_{n+1}\inxeq \verA_n \dR$ (or equivalently $\verA_{n+1}\lexeq \verA_n \dR$).
\label{it:cor-ch-ord}
\end{enumerate}

A correct chain is \emph{witnessing} for a set $X\subseteq\Tt$ if for infinitely many $n$ we have $\verA_n \in X$ and $\verA_{n+1}\inxeq \verA_n \dL$.
\end{definition}

\begin{lemma}
A set $X\subseteq \Tt$ belongs to $\IFinx$ if and only if there exists a correct chain witnessing for $X$.
\end{lemma}

\begin{proof}
First take a correct chain witnessing for $X$. Let $\verB_0,\verB_1,\ldots$ be the subsequence that shows that $(\verA_n)_{n\in\N}$ is witnessing for $X$. In that case, by the definition, for all $n$ we have $\verB_n\in X$ and $\verB_{n+1}\inx \verB_n$ (because $\verB_{n+1}\dM^\omega\lexeq \verB_n\dL\dR^\omega\lex \verB_n\dM^\omega$). Thus, $X$ has an infinite ${\inxeq}$\=/descending chain and belongs to $\IFinx$.

Now assume that $X\in\IFinx$ and $\verB_0\inxg\verB_1\inxg\verB_2\inxg\ldots$ is a sequence witnessing that. Without loss of generality we can assume that $|\verB_{n+1}|>|\verB_n|$ because for each fixed depth $k$ there are only finitely many nodes of $\Tt$ in $\{\dL,\dR\}^{\leq k}$. We can now add intermediate nodes in\=/between the sequence $(\verB_n)_{n\in\N}$ to construct a correct chain witnessing for $X$; the following pseudo\=/code realises this goal:

\begin{minipage}{0.8\linewidth}
\begin{lstlisting}
$n$ := $0$;
$i$ := $0$;
while (true) {
  if ($n$ > $|\verB_i|$) {
    $i$ := $i+1$;
  }
  
  $\verA_n$ := $\verB_i\restr_{n}$; 
  $n$ := $n+1$;
}
\end{lstlisting}
\end{minipage}

Clearly, Property~\ref{it:cor-ch-len} in the definition of a correct chain is guaranteed. Let $i\in\N$ and $n=|\verB_i|$. By the fact that $\verB_{i+1}\inx \verB_i$ we know that $\verB_{i+1}\restr_{n+1}\inxeq \verB_i \dL$. Therefore, for every $n\in\N$ we have $\verA_{n+1}\inxeq \verA_{n} \dR$ and if $n=|\verB_i|$ for some $i$ then $\verA_{n+1}\inxeq \verA_{n} \dL$. It implies that the sequence $(\verA_n)_{n\in\N}$ satisfies Property~\ref{it:cor-ch-ord} in the definition of a correct chain and is witnessing for $X$.
\end{proof}

\section{Hardness for \texorpdfstring{$1$\=/}{1-}counter}

In this section we provide an example of an $\w$\=/language that is $\asigma{1}$\=/complete and can be recognised by a VASS $\Aa_1$ with one counter. $\Aa_1$ is depicted on Figure~\ref{fig:one-vass}, it is very similar to $\Aa_2$. Let $A_0\eqdef\{{<}, d, {|}, i, {+}, {-}, {>}\}$ and let the alphabet $A\eqdef A_0\cup\{\sharp\}$.

\begin{figure}
\centering
\begin{tikzpicture}[->,>=latex,auto,initial text={},scale=0.9]
\tikzstyle{trans}=[scale=0.8]
\node[state,initial  ] (q0) at (0,0) {$q_0$};
\node[state          ] (q1) at (3,0) {$q_1$};
\node[state          ] (q2) at (6,0) {$q_2$};
\node[state,accepting] (qf) at (9,+1) {$q_a$};
\node[state          ] (qr) at (9,-1) {$q_r$};
\node[state          ] (q3) at (12,0) {$q_3$};

\path (q0) edge [loop, in=60 , out=120, looseness=8] node[trans] {$A_0$} (q0)
           edge node[trans] {${<}$} (q1)
      (q1) edge [loop, in=60 , out=120, looseness=8] node[trans] {$d:-1$} (q1)
           edge node[trans] {${|}$} (q2)
      (q2) edge [loop, in=60 , out=120, looseness=8] node[trans] {$i:+1$} (q2)
           edge node[above, trans] {${+}$} (qf)
           edge node[below, trans] {${-}$} (qr)
      (qf) edge node[above, trans] {${>}$} (q3)
      (qr) edge node[below, trans] {${>}$} (q3)
      (q3) edge [loop, in=60 , out=120, looseness=8] node[trans] {$A_0$} (q3)
           edge [in=-90, out=270, looseness=0.8] node[below, trans] {$\sharp$} (q0);

\end{tikzpicture}
\caption{The VASS $\Aa_1$ with one counter that recognises a $\asigma{1}$\=/complete $\w$\=/language.}
\label{fig:one-vass}
\end{figure}

\begin{proposition}
\label{pro:hardness-one}
There exists a continuous reduction from $\IFinx$ to the $\w$\=/language recognised by $\Aa_1$.
\end{proposition}

Similarly as before, we will use the notion of phases and blocks. Since there is only one counter now (and only two letters modifying its value $d$ and $i$) we exchange the definition~\eqref{eq:block-two} by the following one (for $n,m\in\N$ and $s\in\{{+},{-}\}$):
\begin{equation}
B^s(-n,+m)\eqdef {<}\ d^n\ {|}\ i^m\ s\ {>}\ \in A_0^\ast.
\label{eq:block-one}
\end{equation}

Similarly as before, we will take a set $X\subseteq \Tt$ and construct an $\w$\=/word $\infA(X)$. This $\w$\=/word will be a concatenation of infinitely many phases $\finA_0\sharp\finA_1\sharp\cdots$. The $n$\=/th phase $\finA_n$ will depend on $X\cap\{\dL,\dR\}^n$. The configurations $(q_0,c)$ reached at the beginning of an $n$\=/th phase will be in correspondence with nodes $\verA\in\{\dL,\dR\}^n$. The bigger the value $c$, the higher in the DFS order (or the lexicographic order, as they overlap here) the respective node $\verA$ is.

To precisely define our $\w$\=/word $\infA(X)$ we need to define a fast\=/growing functions: $\fun{m}{\N}{\N}$ and $\fun{e}{\Tt}{\N}$:
\begin{align*}
m(-1)&= 1,\\
m(n)&= m(n-1)\cdot 2^n,\\
e(\verA) &= m(|\verA|-1)\cdot b(\verA) \text{ for $\verA\in\Tt$}.
\end{align*}

Notice the following two invariants of this definition, for $n\in\N$ and $\verA,\verA'\in\{\dL,\dR\}^n$:
\begin{align}
\verA\inx \verA' &\Longleftrightarrow e(\verA)\leq e(\verB),\label{eq:order}\\
e(\verA) + m(|\verA|-1) &\leq m(|\verA|).\label{eq:upper}
\end{align}

We take any $n=0,1,\ldots$ and define the $n$\=/th phase $\finA_n$. Let $\finA_n$ be the concatenation of the following blocks, for all $\verA\in\{\dL,\dR\}^n$ and $d\in\{\dL,\dR\}$:
\[B^s\big({-}e(\verA), {+}e(\verA d)\big),\]
where $s={+}$ if $\verA\in X$ and $d=\dL$; otherwise $s={-}$. Thus, the $n$\=/th phase is a concatenation of $2^{n+1}$ blocks, one for each node $\verA d$ in $\{\dL,\dR\}^{n+1}$.

To conclude the proof of Proposition~\ref{pro:hardness-one} it is enough to prove the following two lemmas.

\begin{lemma}
\label{lem:chain-to-run}
If there exists a correct chain witnessing for $X$ then $\Aa_1$ accepts $\infA(X)$.
\end{lemma}

\begin{lemma}
\label{lem:run-to-chain}
If $\Aa_1$ accepts $\infA(X)$ then there exists a correct chain witnessing for $X$.
\end{lemma}

\begin{proofof}{Lemma~\ref{lem:chain-to-run}}
Consider a correct chain $(\verA_n)_{n\in\N}$ witnessing for $X$. Assume that $I\subseteq \N$ is an infinite set such that for $n\in I$ we have $\verA_n\in X$ and $\verA_{n+1}\inxeq \verA_n \dL$. Let us construct inductively a run $\runA$ of $\Aa_1$ on $\infA(X)$. The invariant is that for each $n\in\N$ the configuration of $\runA$ before reading the $n$\=/th phase of $\infA(X)$ is of the form $(q_0,c_n)$ with $c_n\geq e(\verA_n)$. To define $\runA$ it is enough to decide which block to choose from an $n$\=/th phase of $\infA(X)$:
\begin{itemize}
\item if $n\in I$ then choose the block $B^{+}\big({-}e(\verA_n), {+}e(\verA_n \dL)\big)$,
\item otherwise choose the block $B^{-}\big({-}e(\verA_n), {+}e(\verA_n \dR)\big)$.
\end{itemize}
Notice that by the invariant, it is allowed to choose the respective blocks as $c_n\geq e(\verA_n)$. Because of~\eqref{eq:order} and the fact that $(\verA_n)_{n\in\N}$ is a correct chain, the invariant is preserved. As the set $I$ is infinite, the constructed run chooses an accepting block infinitely many times and thus is accepting.
\end{proofof}

\begin{proofof}{Lemma~\ref{lem:run-to-chain}}
Assume that $\runA$ is an accepting run of $\Aa_1$ over $\infA(X)$. For $n=0,1,\ldots$ let $(q_0,c_n)$ be the configuration in $\runA$ before reading the $n$\=/th phase of $\infA(X)$ and assume that $\runA$ chooses a block of the form $B^{s_n}\big({-}e(\verA_n), {+}e(\verA_n d_n)\big)$ in the $n$\=/th phase of $\infA(X)$. Our aim is to show that $(\verA_n)_{n\in\N}$ is a correct chain witnessing for $X$. First notice that by the construction of $\infA(X)$ we have $|\verA_n|=n$.

Clearly, as the counter needs to be non\=/negative, we have $e(\verA_n)\leq c_n$. Notice that by~\eqref{eq:upper} we obtain inductively for $n=0,1,\ldots$ that $c_n< m(n)$. Therefore, we have
\begin{align*}
m(n)\cdot b(\verA_{n+1}) = e(\verA_{n+1}) &\leq c_{n+1} =\\
= c_n - e(\verA_n) + e(\verA_n d_n) &< m(n)+e(\verA_nd_n) =\\
&=m(n)+m(n)\cdot b(\verA_nd_n).
\end{align*}
By dividing by $m(n)$ we obtain $b(\verA_{n+1})< 1+b(\verA_nd_n)$, thus $b(\verA_{n+1})\leq b(\verA_nd_n)$ and $\verA_{n+1}\inxeq \verA_nd_n\inxeq \verA_n\dR$. Moreover, whenever $s_n={+}$ (i.e.~the $n$\=/th chosen block is accepting) then $\verA_n\in X$ and $d_n=\dL$. Therefore, as $\runA$ chooses infinitely many accepting blocks, $(\verA_n)_{n\in\N}$ is witnessing for $X$.
\end{proofof}

This concludes the proof of Proposition~\ref{pro:hardness-one}.

\section{Concluding remarks and related work}

The core result of this paper is a technique of encoding a $\asigma{1}$\=/complete language in a \emph{monotone} way --- Proposition~\ref{pro:hardness-one}.

The results of this paper were obtained independently from the recent results of Finkel in which a family of $4$-counter blind B\"uchi automata is exhibited. The languages recognised by the automata from that family occupy exactly the same levels of the Wadge hierarchy as non-deterministic B\"uchi Turing machines (in particular there are such languages that are $\asigma{1}$\=/complete). The two results are incomparable, as the examples of Finkel involve $4$ counters but instead inhabit much more Wadge levels than only~$\asigma{1}$.

\bibliographystyle{alpha}
\bibliography{mskrzypczak}

\begin{thebibliography}{Kec95}

\bibitem[FS14]{finkel_top_det_petri}
Olivier Finkel and Micha{\l} Skrzypczak.
\newblock On the topological complexity of w-languages of non-deterministic
  {Petri} nets.
\newblock {\em Inf. Process. Lett.}, 114(5):229--233, 2014.

\bibitem[Kec95]{kechris_descriptive}
Alexander Kechris.
\newblock {\em Classical descriptive set theory}.
\newblock Springer-Verlag, New York, 1995.

\end{thebibliography}

\end{document}